\documentclass[preprint,12pt]{article}

\usepackage{amssymb}
\usepackage{hyperref}

\usepackage[latin1]{inputenc}
\usepackage{subfigure}
\usepackage{tikz}
\usepackage{dsfont,color}
\usepackage{tabularx}
\usepackage{amsthm}

\def\R{{\cal R}}
\def\hP{{\hat{\Phi}}}

\newtheorem{theo}{Theorem} 
\newtheorem{lem}{Lemma}
\newtheorem{defi}{Definition}
\newtheorem{prop}{Proposition}
\newtheorem{claim}[prop]{Claim}
\newtheorem{cor}{Corollary}

\title{Fiedler Vectors and Elongation of Graphs: A Threshold Phenomenon on a Particular Class of Trees }

\author{Julien Lef\`evre\\
Aix-Marseille Universit\'{e}, Marseille, France\\
LSIS, UMR CNRS 7296, Marseille, France\\
julien.lefevre@univ-amu.fr}

\date{}

\begin{document}

\maketitle


\begin{abstract}
Let $G$ be a graph. Its laplacian matrix $L(G)$ is positive and we consider eigenvectors of its first non-null eigenvalue that are called Fiedler vector. They have been intensively used in spectral partitioning problems due to their good empirical properties. More recently Fiedler vectors have been also popularized in the computer graphics community to describe elongation of shapes. In more technical terms, authors have conjectured that extrema of Fiedler vectors can yield the diameter of a graph. In this work we present (FED) property for a graph $G$, i.e. the fact that diameter of a graph can be obtain by Fiedler vectors. We study in detail a parametric family of trees that gives indeed a counter example for the previous conjecture but reveals a threshold phenomenon for (FED) property. We end by an exhaustive enumeration of trees with at most 20 vertices for which (FED) is true and some perspectives.
\end{abstract}

\section{Introduction}
\label{sec:in}

Given a undirected graph $G=(V,E)$ where $V=\{v_i \}_{i=1...n}$ are the vertices and $E=\{e_{i,j} \}_{i,j=1...n}$ the edges. The adjacency matrix $A$ is defined by $A_{i,j}=1$ if $i\neq j$ and $e_{i,j} \in E$. $A_{i,j}=0$ otherwise. The degree matrix $D$ is a diagonal matrix where $D_{i,i}=\deg(v_i):=\sum_{j=1...n} A_{i,j}$. The graph laplacian is the matrix $L(G):=D-A$. Since seminal works by Fiedler \cite{fiedler1973algebraic} there have been a considerable amount of theoretical results on spectral properties of graph laplacian.  We will recall a few ones.\\

$-L$ is a symmetric, positive matrix. It has $n$ eigenvalues $0=\lambda_1\leq \lambda_2 \leq ... \leq \lambda_n$. The multiplicity of $0$ equals the number of connected components of $G$. In this article we will focus more precisely on the second smallest eigenvalues, the algebraic connectivity $\alpha(G)$, and an associated eigenvector or Fiedler vector $\Phi$. We have classical bounds on the algebraic connectivity and especially if $G$ is a tree : 
$$ 2\Big(1-\cos\frac{\pi}{n} \Big)\leq \alpha(G)\leq 1$$ with equality if and only if the graph is a path (left) or a star (right)respectively.\\ 

There have been a lot of works on algebraic connectivity of graphs and trees since emergence of this measure \cite{fiedler1973algebraic,grone1987algebraic}. The Fiedler vector has also been intensively studied for instance in partitioning problems for graphs and their applications \cite{simon1991partitioning,alpert1999spectral}. Recently it has been pointed out in an applicative context that the Fiedler vector could yield the diameter of particular graphs \cite{chung2011hot,lefevre2012} even if it has been noticed previously that Fiedler vector could describe elongation of meshes \cite{levy2006}. More precisely, for a closed, smooth and simply connected surface with no holes, it has been conjectured that the extremal points of the second eigenfunction of Laplace-Betrami operator (i.e. the Fiedler vector in a continuous setting) were the more distant points on the surface. In \cite{chung2011hot,lefevre2012} the authors explain also the link with "hot-spots" conjecture that states that Fiedler vector of an open connected subset of $\mathbb{R}^d $ has its extrema on the borders \cite{banuelos1999hot}. The "hot-spots" conjecture is not true in general \cite{burdzy1999counterexample} but it is still an open challenge to characterize extrema of a Fiedler vector on a surface.
A discrete counter-example of the conjecture in \cite{chung2011hot} has been proposed in \cite{evans2011fiedler} with numerical simulations on some examples.\\

In this paper we propose a generalization of numerical results in \cite{evans2011fiedler} and an analytical proof of them. Moreover we emphasize threshold behaviors on a specific class of trees with three parameters, called Rose trees. To our knowledge it is a first attempt to determine quantitatively when Fiedler Vector can describe elongation of a graph.\\

\section{Definitions, notations and first lemmas}

We introduce first our key property ($FED$ for Fiedler Extrema Diameter).

\begin{defi}
Given a graph $G$ whose Fiedler vector $\Phi$ is unique up to a multiplicative constant, we will say that $G$ satisfies the property (FED) i.f.f. 
\begin{enumerate}
\item $\Phi$ has only two extrema.
\item Denoting $m=\arg \min_{i} \Phi_i$, $M=\arg \max_{i} \Phi_i$, the distance between $m$ and $M$ is equal to the diameter of $G$.
\end{enumerate}

\end{defi}

\begin{defi}
Given three integers $s,t,p$ we call Rose tree $\R(s,t,p)$ the graph built from a path of length $s+t+1$ and a star with $p$ branches by connecting the node $s+1$ of the path to the center of the star.
\end{defi}

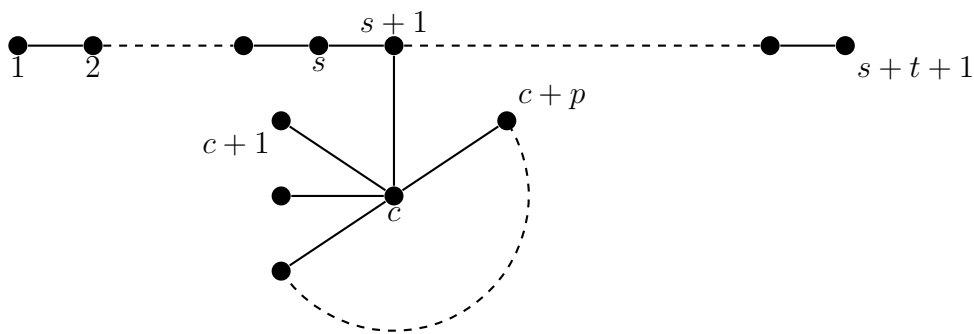
\vspace{1cm}
 \begin{tikzpicture}[style=thick]
 \node[below] at (0,0) {$1$};
  \node[below] at (1,0) {$2$};
   \node[below] at (4,0) {$s$};
    \node[above] at (5,0) {$s+1$};
     \node[below right] at (11,0) {$s+t+1$};
      \node[below] at (5,-2) {$c$};
       \node[below left] at (3.5,-1) {$c+1$};
       \node[above right] at (6.5,-1) {$c+p$};
 \node[circle,fill=black,scale=0.65] (A) at (0,0)  {};  \node[circle,fill=black,scale=0.65] (B) at (1,0) {} ;
  \node[circle,fill=black,scale=0.65] (C) at (3,0) {} ;  \node[circle,fill=black,scale=0.65] (D) at (4,0) {} ;
    \node[circle,fill=black,scale=0.65] (E) at (5,0) {} ;  \node[circle,fill=black,scale=0.65] (F) at (10,0) {} ;
        \node[circle,fill=black,scale=0.65] (G) at (11,0) {} ;  \node[circle,fill=black,scale=0.65] (H) at (5,-2) {} ;
          \node[circle,fill=black,scale=0.65] (I) at (3.5,-1) {} ;  \node[circle,fill=black,scale=0.65] (J) at (6.5,-1) {} ;
           \node[circle,fill=black,scale=0.65] (K) at (3.5,-2) {} ;  \node[circle,fill=black,scale=0.65] (L) at (3.5,-3) {} ;
 \draw[black] (A) -- (B); \draw[black,dashed] (B) -- (C); \draw[black] (C) -- (D);
  \draw[black] (D) -- (E); \draw[black,dashed] (E) -- (F); \draw[black] (F) -- (G);
   \draw[black] (E) -- (H); \draw[black] (H)--(I); \draw[black] (H)--(J);
   \draw[black] (H)--(K); \draw[black] (H)--(L); \draw[black,dashed] (L) arc (-146:33:1.80);
 
 \end{tikzpicture}

 \vspace{1cm}
 

In the following we will only consider Rose trees with $s,t \geq 3$. Namely the diameter of these trees equals $s+t$ and extremal points are precisely vertices $1$ and $s+t+1$. We denote $\alpha(s,t,p)$ the algebraic connectivity of $\R(s,t,p)$. Rose trees can be seen as an hybrid form between stars and paths and natural questions emerge about behavior of algebraic connectivity with respect to parameters $(s,t,p)$.\\

The main result of this article can be summarized by these theorem:

\begin{theo}
\label{theo1}
$\R(s,t,p)$ has an unique Fiedler vector up to a multiplicative constant. There exists two functions $f_m, f_M : \big(\mathds{N}\backslash \{0,1,2\}\big)^2 \rightarrow \mathds{R} $ such as:\\
For $p \leq f_m(s,t)$, (FED) is true for $\R(s,t,p)$.\\
For $p> f_M(s,t)$, (FED) is false for $\R(s,t,p)$.\\
\end{theo}

\begin{theo}
\label{theo2}
Moreover we have more specific results
\begin{eqnarray}
f_m(s,s)=f_M(s,s) \sim_{+\infty} \frac{4}{\pi^2} s^2 \\
\sup_t f_M(s,t) < +\infty
\end{eqnarray}
\end{theo}

The previous inequality tells us that for $p > \sup_t f_M(s,t)  $, even if we take $t$ as big as we want, $\R(s,t,p)$ will not have the behavior of a path in terms of property $(FED)$. \\


Then we begin by some useful results for the following.
\begin{lem} 
\label{connecdecrease}
\cite{fallat2003maximizing}.
Let $G$ a graph. Let $\tilde{G}$ be the graph obtained from $G$ by adding a pendant vertex to a vertex of $G$. Then $\alpha(\tilde{G}) \leq \alpha(G)$. 
\end{lem}

\begin{theo}
\cite{fiedler1975property,grone1987algebraic}
Let $T$ a tree and $\Phi$ a Fiedler vector. Two cases can occur:
\begin{itemize}
\item[A] All values $\Phi_i$ are different from zero. Then $T$ contains exactly one edge $(p,q)$ such that $\Phi_p >0$ and $\Phi_q<0$. The values in vertices along any path starting from $p$ (resp. $q$) and no containing $q$ (resp. $p$) increases (resp. decreases). $T$ is said to be of type II.
\item[B] The set $N_0=\{i | \Phi_i=0 \} $ is non-empty. Then the graph induced by $T$ on $N_0$ is connected and there is exactly one vertex $v$, called characteristic vertex, in $N_0$ having one neighbour not belonging to $N_0$. The values along any path in $T$ starting from $v$ are increasing, decreasing or zero. $T$ is said to be of type I.
\end{itemize}
\end{theo}

Hence
\begin{cor}
\label{corrol}
A Fiedler vector of a tree attains extremal values at some of its leaves (also called pendant vertices).
\end{cor}

%

Now let us come back to $\R(s,t,p)$. Writing $c=s+t+2$,  we have the relation:
\begin{equation}
\label{eqcenter}
\Big(1-\alpha(s,t,p)\Big)\Phi_i=\Phi_c \hspace{1cm} i=s+t+3, ... , s+t+p+2 
\end{equation}
When $s\neq 0$ and $t\neq 0$, $\alpha(s,t,p)<1$  and all the $\Phi_i$ are equal. We denote $\hat{\Phi} $ the common value.\\ From the Corollary \ref{corrol} it follows that possible extrema of a Fiedler vector are $\Phi_1$, $\Phi_{s+t+1}$ and $\hP$.\\

There exist classical formulas to obtain the characteristic polynomial of a tree \cite{brouwer2012spectra} but in our case we adopt a more local strategy that allows to obtain relations on first eigenvalue and associated eigenvector as well. With the following equalities we will be able to quantify precisely the extremal values of $\Phi_i$.

\begin{claim}
\label{prop1}
Writing $\alpha=\alpha(s,t,p)$ we have the following relations:
\begin{eqnarray}
(1-\alpha)\hP& =& \Phi_c \\
P(\alpha) \hP&= &\Phi_{s+1} \label{eq1}\\
- Q(\alpha) \hP&=&\Phi_s + \Phi_{s+2}  \label{eq2}\\
R_{s}(\alpha) \Phi_i&=& R_{i-1}(\alpha) \Phi_{s+1} \hspace{1cm} i=1,...,s  \label{eq3}\\
R_{t}(\alpha) \Phi_{i+s+1}&=& R_{t-i}(\alpha) \Phi_{s+1} \hspace{1cm} i=1,...,t  \label{eq4}
\end{eqnarray}
where $P$, $Q$ and $R_i$ are polynomials:
 \begin{eqnarray}
 P(X)&=&X^2-(p+2)X+1\\
 Q(X)&=&(X-3)P(X)+(1-X)\\
 R_0(X)&=&1\\ 
 R_1(X)&=&1-X\\
 R_n(X)&=&(2-X)R_{n-1}(X)-R_{n-2}(X) \hspace{1cm} n\geq2 \label{reccheb}
 \end{eqnarray}
\end{claim}

\begin{proof}

We have first:
$$(p+1-\alpha) \Phi_c=\Phi_{s+1}+p\hP$$
so by using Equation \ref{eqcenter} we obtain the relation between $\Phi_{s+1}$ and $\hP$.
Next we have $$(3-\alpha) \Phi_{s+1}=\Phi_s+\Phi_{s+2}+\Phi_c$$ and by writing $\Phi_c$ and $\Phi_{s+1}$ in function of $\hP$ we obtain the Equation \ref{eq2}.\\
Then we consider the branch $s$ for instance (it is the same proof for branch $t$ by adapting the indices). We have first $(1-\alpha)\Phi_1=\Phi_2$. Then $(2-\alpha)\Phi_2=\Phi_1+\Phi_3$ which yields $\Phi_3=\big((2-\alpha)(1-\alpha)-1\big)\Phi_1$. We can show by a simple recurrence that $\Phi_i=R_{i-1}\Phi_1$ for $i=1,..,s+1$ where $R_i$ are defined by the recurrence relation \ref{reccheb}. Then multiplying by $R_s$ and since $\Phi_{s+1}=R_s\Phi_1$ we obtain Equation \ref{eq3}.
The recurrence relation \ref{reccheb} can also be found for instance in \cite{guo2008conjecture}.
 \hfill$\blacksquare$
\end{proof}

\begin{claim}
\label{prop2}
$\alpha(s,t,p)$ is the first non-null root of the polynomial:
\begin{equation}
\chi_{p,s,t}:=(R_sR_{t-1}+R_{s-1}R_t)P+QR_sR_t
\end{equation}

\end{claim}

\begin{proof}
Taking $i=s$ (resp $i=1$) in equality \ref{eq3} (resp \ref{eq4}) we get:
\begin{eqnarray*}
 R_s(\alpha) \Phi_s&=& R_{s-1}(\alpha) \Phi_{s+1} \\
 R_t(\alpha) \Phi_{s+2}&=& R_{t-1}(\alpha) \Phi_{s+1}
 \end{eqnarray*}
With \ref{eq1} the two previous equations depend only on $\hP$ and we can sum them by multiplying the first equation by $R_t(\alpha)$ and the second by $R_s(\alpha)$. We conclude thanks to equality \ref{eq2} of Proposition \ref{prop1}. \hfill$\blacksquare$
\end{proof}

\begin{lem}
\label{lem1}
$$
\forall x \in [0,1] \,\, R_n(x)=\frac{\cos(n+1/2)\theta}{\cos \theta/2} \textrm{ with } \cos \theta=1-x/2
$$
\end{lem}

\begin{proof}
The recurrence \ref{reccheb} suggests writing $R_n$ as a combination of Chebyshev polynomials \cite{mason2003chebyshev}. We use the change of variable $y=1-x/2$ from $[0,1] $ to $[1/2,1] $ . We denote $S_n$ the polynomial defined by $S_n(y):=R_n(2(1-y))$ which satisfies the classical relation:
$$n\geq2 \hspace{0.5cm} S_n(y)=2yS_{n-1}(y)-S_{n-2}(y)  \textrm{ and } S_0(y)=1 ,\, S_1(y)=2y-1$$
$S_n$ must therefore be a combination of $T_n$ and $U_n$, Chebyshev polynomials of the first and second kind that satisfies the previous recurrence relation with $T_0=1,\, T_1(x)=x$ and $U_0=1,\, U_1(x)=2x$.
So we can obtain:
$$S_n(y)=\frac{1}{y}\bigg(T_n(y)-U_n(y)\bigg)+U_n(y) $$
But we also know that $T_n=U_n-xU_{n-1} $ which yields
$$S_n(y)=\frac{\sin (n+1)\theta-\sin n\theta}{\sin \theta} \textrm{ with } \cos \theta=y$$
This last expression can be simplified again with simple trigonometric formula. \hfill$\blacksquare$
\end{proof}

We denote 
$$r(s):=2\bigg(1-\cos \frac{\pi}{2s+1}\bigg) $$
the first positive root of $R_s$.

\begin{lem}
\label{lemR_s}
On $]0,r(s)[$, we have $R_s' <0$. There exists a constant $\beta_s>0$ such as $R_s''>0$ on $[0,\beta_s[$. Moreover $R_s'(0)=-s(s+1)/2$  and $\chi_{p,s,t}'(0)=p+s+t $
\end{lem}

\begin{proof}
We can derive:
$$R_s'(x)=\theta'(x) \frac{dR_s}{d\theta} \hspace{1cm} \theta'(x)=\frac{1}{2\sin \theta}$$  
A calculation yields
$$ \frac{dR_s}{d\theta}=-\frac{s\sin(s+1)\theta+(s+1)\sin s\theta}{2\cos^2 \theta/2} $$
So we have first that $R_s'(x) <0$.\\ 
Next around $0$:
$$\frac{dR_s}{d\theta}=-2s(s+1)\theta + O(\theta^2) \hspace{1cm}  \theta'(x)=\frac{1}{2\theta}+O(1) $$
which yields
$$R_s'(0)=-s(s+1)/2$$
and allows to obtain $\chi_{p,s,t}'(0)=p+s+t $.\\
Then by the chain rule again:
$$ R_s''(x)=\bigg(F(\theta(x)) \bigg)'=\theta'(x)  \frac{dF}{d\theta}$$
with $F(\theta)=\frac{1}{2 \sin \theta}\frac{dR_s}{d\theta}$.\\
After a calculation we get:
\begin{eqnarray*}
\frac{dF}{d\theta}(\theta)&=&-\big(16\sin^2\theta/2\cos^4\theta/2\big)^{-1} \Bigg(\frac{s(s-1)}{2} \sin (s+2)\theta + \bigg(\frac{s(s+1)}{2} -1\bigg)\sin (s+1)\theta \\
& - &   \bigg(\frac{s(s+1)}{2} -1\bigg)\sin s\theta - \frac{(s+1)(s+2)}{2} \sin (s-1)\theta \Bigg)
\end{eqnarray*}
Then we can compute the Taylor expansion of the numerator around $0$ which equals $-4(s-1)s(s+1)(s+2) \theta^3 + O(\theta^5)$. It means that we can find a strictly positive constant $\theta_s$ such as $\frac{dF}{d\theta} <0$ on $[0,\theta_s]$. Therefore taking $\beta_s=2(1-\cos \theta_s)$ we get the desired result.\hfill$\blacksquare$

\end{proof}

In the next two parts we describe threshold properties of $\R(s,t,p)$, first in the case where $s=t$ and secondly when $t>s$. We have to remark that our proof is quite elementary even if technical and does not require theorems for determining characteristic vertices through Perron branches \cite{kirkland1996characteristic} which lead to technical developments as well.


\section{Analysis of Rose trees when $t=s$}


In this section we study $\R(s,t,p)$ when $t=s$ and we denote $\alpha(s,p)=\alpha(s,s,p)$. One can simplify:
\begin{equation}
\label{chiRs}
\chi_{p,s,s}=R_s \Big( 2R_{s-1}P+R_sQ \Big)
\end{equation}

We start by this basic result:

\begin{lem}
\label{cosirra}
$P(r(s))\neq 0$ 
\end{lem}

We can find a short proof in \cite{jahnel2010co}.

\begin{lem}
\label{lem2}
As soon as $p> \big(r(s)-1\big)^2/r(s)$,  $\alpha(s,p)$ is a root of $f_p:=2R_{s-1}P+R_sQ$ or in other terms 
$$\alpha(s,p) < r(s) $$
\end{lem}

\begin{proof}
%
%
We have $R_{s}(0)=1$, $P(1/p) \rightarrow 0 $ and $Q(1/p) \rightarrow 1$ when $p \rightarrow +\infty$.
So from the definition of $f_p$ we can see that $f_p\big(1/p\big) \rightarrow 1$ when $p \rightarrow +\infty$. Moreover 
$$f_p\big(r(s) \big)= 2R_{s-1}\big(r(s)\big)P\big(r(s)\big) $$ Since $R_{s-1}\big(r(s)\big)>0$ we can see that $f_p\big(r(s)\big) \rightarrow -\infty$ when $p \rightarrow +\infty$. So for $p$ large enough there exists a root of $f_p$ in $]1/p,r(s) [ $. This root is also a root of $\chi_{p,s}$ so we have necessarily $\alpha(s,p) < r(s)$. We can refine the "p large enough" by saying that as soon as $P\big(r(s)\big)<0$ the previous inequality is strict. This yields 
$$ p > \frac{\big(r(s)-1\big)^2}{r(s)} $$ \hfill$\blacksquare$

\end{proof}

\begin{lem}
\label{lem2bis}
We define 
$$h_s(X):= 2R_{s-1}+(X-3)R_s$$
$h_s$ is increasing on $[0,r(s)]$ and has an unique root in $]0,r(s)[$ .
\end{lem}

\begin{proof}
We write $h_s$ in term of $\theta$ as in Lemma \ref{lem1}. In the full expression, we have a common factor $\big(\cos \theta/2 \big)^{-1}$ which is increasing with $\theta$ and positive. So we can only study the monotony of the numerator  that we derive:
$$\big(h_s(\theta)\cos \theta/2 \big)'=-(2s-1) \sin (s-1/2)\theta + (1+2\cos \theta )(s+1/2)\sin (s+1/2)\theta +2\sin \theta \cos (s+1/2) \theta $$
We transform the first term:
$$  -(2s-1)\sin (s+1/2-1)\theta =-(2s-1) \bigg( \sin (s+1/2)\theta \cos \theta  -\cos (s+1/2)\theta \sin\theta \bigg) $$
and we obtain:

$$\big(h_s(\theta)\cos \theta/2 \big)'= 2 \cos \theta \sin(s+1/2)\theta + (s+1/2)\sin(s+1/2)\theta +(2s+1)\cos(s+1/2)\theta \sin \theta$$

It is easy to check that it is strictly positive on $[0,\pi/(2s+1)]$. So $h_s$ is increasing on $[0,r(s)]$.

From $h_s(0)=-1$ and $h_s(r(s))=2R_{s-1}(r(s))>0$ we get the last part of the lemma. \hfill$\blacksquare$

\end{proof}

\begin{prop}
\label{prop3}
For $s$ fixed, $\alpha(s,p)$ is decreasing with $p$ and converges to $L(s) > 0$.
\end{prop}

\begin{proof}
 
%
 
 If we remove one edge linking the center $c$ to an isolated branch, we simply obtain the graph $\R(s,t,p-1)$ and therefore $\alpha(s,p-1)\leq \alpha(s,p)$ thanks to Lemma \ref{connecdecrease}. Since $\alpha(s,p)$ is bounded, it has a limit $L(s) \geq 0$.\\
 
 Next we have:
 \begin{equation}
\label{f_p}
 f_p(X)=P(X)h_s(X)-(X-1)R_s(X) \end{equation}
 and by definition we have $\chi_{p,s,s}=R_s f_p$. So $0$ is a root of $f_p$ since it is not a root of $R_s$. Its multiplicity is $1$ since the graph is connected. Therefore we have:
\begin{equation}f_p=X \big(A_s(X) p + B_s(X) \big) \end{equation}
 where $A_s$ and $B_s$ are polynomials that do not depend on $p$. 
 Since $f_p(\alpha(s,p))=0$ and $\alpha(s,p)\neq0$ we must have $A_s((\alpha(s,p)))p+B_s((\alpha(s,p)))=0 $. Dividing by $p$ and making $p \rightarrow +\infty$ it imposes $A_s(L)=0$.
 From Eq. \ref{f_p} we can obtain that $A_s=-h_s$.  So $h_s(L)=0$ and $L\leq r(s)$ from Lemma \ref{lem2}. Therefore, from Lemma \ref{lem2bis} $L$ is the unique root of $h_s$ on  $[0,r(s)]$. It implies that $\alpha(s,p) $  converges to $L(s) \in ]0,r(s)[$. \hfill$\blacksquare$

\end{proof}


\begin{lem}
\label{lem2ter}
For $p\leq \big(r(s)-1\big)^2/r(s)$,  we have $\alpha(s,p)=r(s)$ 
\end{lem}
\begin{proof}
We consider two cases:\\

1) If $f_p$ has no root in $]0,1] $ then $\alpha(s,p) = \min\big(r(s),\rho(s,p)\big)=r(s)$.\\

2) In the other case we call $\rho(s,p)$ the first  root of $f_p$ in $]0,1] $. We assume $\rho(s,p)<r(s)$ and we will show that it leads to an absurdity.\\ 
First, since $P$ is decreasing we have $P\big(\rho(s,p)\big) > P\big(r(s)\big)$. And $P\big(r(s)\big) \geq 0$ since $p\leq \big(r(s)-1\big)^2/r(s)$.\\
Secondly, we can rewrite:
 $$h_s P=2R_{s-1} P + (X-3)PR_s = f_p+(X-1)R_s$$
 so we have the following equality:
 $$h_s\big(\rho(s,p)\big)  P\big(\rho(s,p)\big) =R_s\big(\rho(s,p)\big) \big(\rho(s,p)-1\big)  $$
 Since $P\big(\rho(s,p)\big) > 0 $, $\rho(s,p)<1$ and $R_s\big(\rho(s,p)\big) > 0$ we obtain $h_s\big(\rho(s,p)\big) < 0$. But from Lemma \ref{lem2bis} this implies $\rho(s,p)< L(s)$. At last $\alpha(s,p) = \min\big(r(s),\rho(s,p)\big)$ so $\alpha(s,p) < L(s)$. But we can find $p' \geq p$ and satisfying $p' \geq  \big(r(s)-1\big)^2/r(s)$ for which $\alpha(s,p') \leq \alpha(s,p) < L(s) $ which contradicts Proposition \ref{prop3}. Therefore we have proved that $\rho(s,p)\geq r(s)$ which implies that $\alpha(s,p) = \min\big(r(s),\rho(s,p)\big)=r(s)$. \hfill$\blacksquare$

\end{proof}

%

%

\begin{prop}
\label{uniquess}
Fiedler vector of $\R(s,s,p)$ is unique up to a multiplicative constant
\end{prop}

\begin{proof}
We consider two cases:\\

First if $p > (r(s)-1)^2/r(s)$ then $R_s(\alpha)\neq 0$ from Lemma \ref{lem2}. Necessarily $\Phi_{s+1} \neq 0$ else all $\Phi_i$ will be equal to zero from Claim \ref{prop1}. So all the values $\Phi_i$ for $i=1,...,2s+1$ are uniquely determined from Equations \ref{eq3} and \ref{eq4}. Last $P(\alpha) \neq 0$ else $\chi_{p,s,s}(\alpha)=R_s(\alpha)^2(1-\alpha) $ would equal zero. So all the values of $\Phi$ are uniquely determined by $\Phi_{s+1}$.\\

Secondly if $p \leq (r(s)-1)^2/r(s)$ we have from Lemma \ref{lem2ter} $R_s(\alpha(s,p))=0$. Then from Claim \ref{prop1} we get for any $1\leq i\leq s$, $R_{i-1}(\alpha(s,p))\Phi_{s+1}=0$. We deduce that $\Phi_{s+1}=0$, $\Phi_s=-\Phi_{s+2}$ and since $P(\alpha(s,p))\neq 0$ we obtain $\hat{\Phi}=0$. And since $R_{s-1}(r(s)) \neq 0$, $\Phi_i$ are uniquely determined by the value $\Phi_s$ from:
$$\Phi_{i} = \Phi_{s} \frac{R_{i-1}(r(s)}{R_{s-1}(r(s))} \hspace{0.5cm} i=1,...,s$$
And necessarily $\Phi_s\neq 0$ elsewhere Fiedler vector will be trivially null. 


 \hfill$\blacksquare$
\end{proof}

So it is now possible to demonstrate the threshold property on the symmetric Rose graph $\R(s,s,p)$.\\

If $p \leq (r(s)-1)^2/r(s)$ then $\R(s,s,p)$ is of type I from the previous proposition and $\hP=0$.
Therefore $\Phi_{1}$ and $\Phi_{2s+1}$ are the two extrema of the Fiedler vector. The property (FED) is satisfied so $f_m(s,s) \geq (r(s)-1)^2/r(s)$.\\

Secondly, if $p > (r(s)-1)^2/r(s)$ $P(\alpha(s,p)) < 0$ and $\hP$ has the opposite sign of $\Phi_1$ and $\Phi_{2s+1}$. The property (FED) is not satisfied. This yields $$ f_M(s,s) \leq (r(s)-1)^2/r(s) $$ and therefore:
$$f_m(s,s)=f_M(s,s):=f(s,s)$$
A direct computation allows to conclude that $$f(s,s) \sim_{+\infty} \frac{4}{\pi^2} s^2 $$ \hfill$\blacksquare$


%



\section{Rose trees when $t > s$}

We consider here that $s < t$. 
By extending the path $1,..,s$ and with the Lemma \ref{connecdecrease} we get:
$$\alpha(t,p)  \leq \alpha(s,t,p)  $$
By contracting the $p$ leaves and the center we obtain:
$$
\alpha(s,t,p) \leq r\big( (s+t)/2\big)
$$
By contracting the $t-s$ vertices of the branch $t$ we get:
$$
\alpha(s,t,p) \leq \alpha(s,s,p) \leq r(s)
$$

\begin{lem}
\label{phispun}
$\Phi_{s+1} \neq 0$ for $\R(s,t,p)$ with $t>s$.
\end{lem}

\begin{proof}
We assume that $\Phi_{s+1}=0$. Then we consider two cases:\\
\begin{itemize}

\item If $P(\alpha)=0$ then  $\chi_{p,s,t}=(1-\alpha)R_s(\alpha)R_t(\alpha)=0$ which implies that $\alpha=r(t)$. But that would imply $P(r(t)=0$ which contradicts Lemma \ref{cosirra}.

\item Then $\hP=0$ then $\Phi_s+\Phi_{s+2}=0$ from Equation \ref{eq2}. 
From Equation \ref{eq3} we obtain $\Phi_i=0$ for $i=1,...,s$. Then $\Phi_{s+2}=0$. From Equation \ref{eq4} and
$$ R_{t-1}(\alpha)\Phi_{s+1+i}= \phi_{s+2} R_{t-i}(\alpha) \hspace{0.5cm} i=1,...,t$$
combined with the fact that $\R_{t-1}(\alpha)$ and $R_t(\alpha)$ can not simultaneously equal zero we deduce that $\Phi_{s+1+i}=0$ for $i=1,...,t$ and the Fiedler vector would be identically null which is absurd.

\end{itemize}
\hfill$\blacksquare$
\end{proof}

\begin{prop}
Fiedler vector of $\R(s,t,p)$ is unique up to a multiplicative constant.
\end{prop}

\begin{proof}

From the previous Lemma $\Phi_{s+1} \neq 0$ then $\hP \neq 0$ and $P(\alpha) \neq 0$ from Equation \ref{eq1}. Since $\alpha<r(s)$ we get from Equation \ref{eq3} that $\Phi_i$ are uniquely determined by $\Phi_{s+1}$ for $i=1,...,s$. Next if $\alpha=r(t)$ then $\chi_{p,s,t}(\alpha)=R_s(\alpha)R_{t-1}(\alpha)P(\alpha)=0$. This would imply that $R_{t-1}(\alpha)=0$ which is impossible. So $R_t(\alpha) \neq 0$ and values $\Phi_{s+1+i}$ are uniquely determined by $\Phi_{s+1}$ from Equation \ref{eq4}.
\hfill$\blacksquare$
\end{proof}

Next we have to compare $ \Phi_1 $, $ \Phi_{s+t+1}$ and $\hP$ which are the three possible extremal values of the Fiedler vector. From Lemma \ref{phispun} and Equation \ref{eq1}, $\hP \neq 0$

We use first a technical lemma:

\begin{lem}
We consider the following function:
 $$g(\theta)=-(p+1)\cos \theta +(p+2)\cos 3\theta-\cos (2s+1)\theta$$ 
 If $p \leq f(s,s)-1$ then $g > 0$ on $[0,\theta_s]$ with $\theta_s= \frac{\pi}{2(2s+1)}$ . 
\end{lem}

\begin{proof}

We compute the third derivative of $g$:
$$ g'''(\theta)=-(p+1) \sin \theta +27(p+2)\sin 3 \theta - (2s+1)^3 \sin (2s+1)\theta $$
then using the classical inequality $\frac{2}{\pi} u \leq \sin u \leq u $ on $[0,\frac{\pi}{2}]$ we can have the following upper bound:
$$ g'''(\theta) \leq \theta \bigg(81(p+2)-\frac{2}{\pi}(p+1)-\frac{2}{\pi} (2s+1)^4 \bigg) $$
From the bound on $p$ and from $\cos u \leq 1-\frac{4u^2}{\pi^2}$ on $[0,\frac{\pi}{3}]$, we obtain
$$ g'''(\theta) \leq \theta \bigg(81\big(f(s,s)+1\big)-\frac{2}{\pi} (2s+1)^4 \bigg) \leq \theta \Bigg(81\bigg( \frac{1}{8} (2s+1)^2 +1 \bigg) -\frac{2}{\pi}(2s+1)^4 \Bigg) $$
The right term is negative as soon as $s\geq 2$.
So $g''$ is decreasing on $[0,\theta_s]$. Moreover 
$$g''(0)=(2s+1)^2-8p+17\geq (2s+1)^2-8f(s,s)+9\geq 9$$
And since $g'(0)=0$ we conclude that $g'$ is increasing and possibly decreasing (depending on the sign of $g''(\theta_s) $). $g'$ vanishes at most once and since $g(0)=0$ the sign of $g(\theta_s)$ tells us wether $g$ vanishes of not:
$$ g(\theta_s)=\cos (3\theta_s )-(p+1)\big(\cos (\theta_s ) - \cos (3\theta_s ) \big) \geq  \cos (3\theta_s )-f(s,s)\big(\cos (\theta_s ) - \cos (3\theta_s ) \big)  $$
And we conclude that $g(\theta_s) \geq 0$ since 
$$u \mapsto 2\cos(3u)\big(1-\cos(2u)\big) -\big(1-2\cos(2u)\big)^2\big(\cos(u)-\cos(3u) \big)=16 \sin^4 u \cos 3 u $$
is positive on $[0,\pi/6] $. \hfill$\blacksquare$

%
%


\end{proof}

\begin{prop}
\label{inequality}
For $s>0$ fixed and $s \leq t $ we have:
 \begin{eqnarray} f(s,s)-1& \leq &f_m(s,t) \\
 f_M(s,t) &\leq& f(t+2,t+2) 
 \end{eqnarray}
 
\end{prop}

\begin{proof}

For the first inequality, we consider $p\leq f(s,s)-1$. Then $P(\alpha(s,t,p))\geq P(r(s)) > 0$. $\Phi_1$ and $\hP$ have the same sign and necessarily $\Phi_{s+t+1}$ has the opposite sign. So if we show that $(P-R_s)(\alpha(s,t,p)) > 0$ we will obtain $|\Phi_1/\hP|>1$ and the property (FED) will be satisfied which will imply $p\leq f_m(s,t)$.

We have $$(P-R_s)(x)=x^2+\Big(1-(p+2)x-R_s(x) \Big) \geq  1-(p+2)x-R_s(x)$$
We can express the right part in term of $\theta= \arccos\Big(1-x/2 \Big) $ and we obtain after some trigonometric calculus
$$\big(\cos \theta/2\big)^{-1}\Big(\cos \theta/2-2(p+2)(1-\cos \theta)\cos \theta/2-\cos (s+1/2)\theta \Big)=\big(\cos \theta/2\big)^{-1} g(\theta/2)$$
And we conclude from the previous Lemma.\\

For the second inequality, we consider $p\geq f(t+2,t+2)$. We have $0\geq P(r(t+2))\geq P(\chi_{p,s,s})$.
We compute $$\chi_{p,s,s}\big(r(t)\big)=P\big(r(t)\big)R_s\big(r(t)\big)R_{t-1}\big(r(t)\big)$$
The two last terms are positive and $P\big(r(t)\big) \leq 0$. But $\chi_{p,s,t}'(0)>0$ from Lemma \ref{lemR_s} and since $\alpha(s,t,p)$ is the first root of $\chi_{p,s,t}$ we have necessarily $\alpha(s,t,p) \leq r(t)$.
So $R_t(\alpha(s,t,p)) \geq 0$ then $\Phi_1$ and $\Phi_{s+t+1}$ have the same sign and necessarily $\hP$ has the opposite sign. The property (FED) is not satisfied which implies $p\geq f_M(s,t)$.\hfill$\blacksquare$
\end{proof}

\begin{prop}
For $s>0$ fixed we have:
 $$\sup_{ t } f_M(s,t)<+\infty$$
\end{prop}

\begin{proof}
$\alpha(s,t,p)$ is decreasing with $t$ and bounded by $r((s+t)/2)$ which converges to $0$ when $t \rightarrow +\infty$.
Then $(P/R_s)(\alpha(s,t,p)) \rightarrow 1$ and for $t$ big enough, $\Phi_1$ and $\hP$ have the same sign. So we compute the Taylor expansion of $P/R_s$ around $0$:
$$ (P/R_s)(x)=\frac{1-(p+2)x+x^2}{1-\Big(s(s+1)/2\Big) x+ O(x^2)}=1-\Big(p+2-s(s+1)/2\Big) x + O(x^2)$$
If $p>s(s+1)/2-2$ there exist $t_{p,s}$ such as for $t>t_{p,s}$, $(P/R_s)(\alpha(s,t,p))<1$
The last case is possibly problematic since $t_{p,s}$ could be a diverging sequence with $p$.\\

So, knowing that $R_s''(x) \geq 0$ on $[0,\beta_s[ $ from Lemma \ref{lemR_s} we can obtain the following upper bound on $[0,\min(\beta_s,\big(s(s+1)\big)^{-1})]$:

$$ (P/R_s)(x) \leq \frac{1-(p+2)x+x^2}{1-\Big(s(s+1)/2\Big) x} $$

We use the following inequality on $[0,1/2]$:

$$\frac{1}{1-u}\leq 1+u+2u^2 $$

which gives us the polynomial approximation:

$$(P/R_s)(x) \leq 1+ a_{p,s} x + b_{p,s} x^2 + c_{p,s} x^3 + d_{s} x^4$$

with $a_{p,s}=s(s+1)/2-(p+2) <0 $, $b_{p,s}$, $c_{p,s}$ are decreasing sequences with $p$ whose analytical formula can be simply obtained (but not presented for clarity).
Considering the case $p(s)=s(s+1)/2-1$ the polynomial bound is
$$ 1- x + b_{p(s),s} x^2 + c_{p(s),s} x^3 + d_{s} x^4 $$
So we can find a value $x_0(s)>0$ such as the previous bound is strictly inferior to 1 in $]0,x_0(s)[ $. And since $a_{p,s}$, $b_{p,s}$, $c_{p,s}$ are decreasing with $p$ we can conclude that for each $p \geq p(s)$ and for each $x \in ]0,x_0(s)[ $ $(P/R_s)(x)<1$. 

We consider $t_s$ such as for $t \geq t_s$, $\alpha(s,t,p(s))<x_0(s)$, then obviously $(P/R_s)(\alpha(s,t,p(s)))<1$.  
Since $\alpha(s,t,p)$ is decreasing with $p$, for $p \geq p_0(s)$ and for $t \geq t_s$ we have also $(P/R_s)(\alpha(s,t,p))<1$.

Therefore, for $t \geq t_s$ and $p \geq p(s)$, the property (FED) is false, so $f_M(s,t)\leq p(s)$. Moreover from Proposition \ref{inequality}:
$$ \max_{s<t \leq t_s} f_M(s,t) \leq  \max_{s<t \leq t_s} f(t+2,t+2) = f(t_s+2,t_s+2) $$

And we get the desired result:

$$ \sup_{s<t} f_M(s,t) \leq \max\Big(\max_{s<t \leq t_s} f_M(s,t) ,  \sup_{t_s<t} f_M(s,t) \Big)=\max\Big(f(t_s+2,t_s+2),p(s)\Big) <+ \infty  $$
\hfill$\blacksquare$
\end{proof}

\paragraph{Remark: } Actually it seems through numerical simulations that we could have finer results for $s\leq t \leq t' $
\begin{eqnarray}
f_m(s,t)&=&f_M(s,t):=f(s,t)\\
f(s,s) &\leq& f(s,t) \leq f(s,t') \leq f(t',t') \\
f(s,t) &\leq& s(s+1)/2-2 
\end{eqnarray}
These numerical results suggest more general properties, that are particularly technical to obtain in our case. Indeed we conjecture the following assertion on trees:\\

Given a tree $T$, we consider $m=\arg \min_{i} \Phi_i$ and $M=\arg \max_{i} \Phi_i$. Then we can consider the two cases:
\begin{enumerate}
\item If (FED) is satisfied for (T) then (FED) is also satisfied for the tree obtained by adding a vertex to $m$ or $M$.
\item If (FED) is not satisfied for (T) then (FED) is also not satisfied for the tree obtained by deleting vertex $m$ or $M$.
\end{enumerate}

\section{Perspectives}

The threshold behavior for $\R(s,t,p)$ can be simply compared to what happens for the Star-like tree  $S(n,\underbrace{1,...,1}_{p})$ obtained by collapsing one end of a path of length $n+1$ to the center a star with $p$ leaves. From Corollary \ref{corrol} extrema of a Fiedler vector gives the diameter of the graph \footnote{Unicity of Fiedler vector in this case can be obtained but it is not a concern here.}. Even Rose trees and Star-like trees are obtained as a "tradeoff" between the star and the path, they do not reveal the same behavior for (FED) property. In particular the way the path is added to  the center or one leave of the star seems to be the explanation of this difference and suggest more general results.\\

Last but not least, natural questions emerge when considering the two previous categories of graph whose behavior is very different with respect to property (FED): How many trees with $n$ vertices satisfy property (FED) ?
We have proposed an exhaustive enumeration of free trees with $n \leq 20$ vertices  thanks to the algorithm in \cite{li1999advantages} \footnote{We have used the C-code proposed on \href{http://theory.cs.uvic.ca/inf/tree/FreeTrees.html}{http://theory.cs.uvic.ca/inf/tree/FreeTrees.html} }. Results are displayed on Tab \ref{tableau}.


 \begin{table}[h!]
 \begin{center}
 \scriptsize
 \begin{tabular}{|c||c|c|c|c|c|c|c|c|c|c|c|c|c|c|c|c|c|c|}

 \hline
$n$ & 11&12&13&14&15&16&17&18&19&20 \\
\hline
\hline
Trees & 235&551&1301&3159&7741&19320&48629&123867&317955&823065 \\
Trees with $\neg (FED)$ &0&1&5&21&72&240&757&2331&7012&20807 \\
Ratio ($\%$) &0.00&0.18&0.38&0.66&0.93&1.24&1.56&1.88&2.21&2.53\\
\hline
 \end{tabular}
   \normalsize
 \end{center}

 \caption{ \label{tableau} Number of free trees with $n$ vertices that do not satisfy (FED) property and their ratio with respect to all free trees with $n$ vertices.}
 \end{table}
 
 We observe that the smallest tree that do not satisfy property (FED) has $12$ vertices and is $\R(3,3,4)$. The proportion of these trees increases slightly with $n$. We can conjecture that the density of such trees converges to a non-null value when $n \rightarrow + \infty$. An interesting question would be to know if this limit is $1$ or something else. Other conjectures could be formulated by observing the links between the property (FED) and classical quantities such algebraic connectivity. For this, exhaustive enumeration are rapidly intractable because the number of free trees with $n$ vertices is asymptotically $O( n^{-5/2} \rho^{-n})$ with $\rho \approx 0.338 $ \cite{otter1948number}. Recent approaches for random sampling of unlabeled combinatorial structures \cite{bodirsky2010boltzmann} would probably be of great help.

\bibliographystyle{elsarticle-num}
\bibliography{biblio_rose}

\end{document}